\documentclass{llncs}

\usepackage{wrapfig}
\usepackage{tikz}
\usepackage{tikz-qtree}
\usepackage{wrapfig}
\usepackage{amssymb,amsmath,stmaryrd}
\usepackage{latexsym}
\usepackage{amsfonts}
\usepackage{layout}
\usepackage{picinpar}
\usepackage{multicol}
\usepackage{times}
\usepackage[lined,boxed,noend,ruled,resetcount]{algorithm2e}
\usepackage{charter}
\usepackage{stmaryrd}
\usepackage{epic}
\usepackage{eepic}
\usepackage{amssymb}
\usepackage{gastex}
\usepackage{wasysym}
\usepackage{comment}

\usepackage{tikz}

\title{Online Firefighting on Grids\thanks{The support of GEO-SAFE,  H2020-MSCA-RISE-2015 project \# 691161 is fully 
acknowledged.}}

\author{Marc Demange\inst{1} \and David Ellison\inst{1} \and Raffaella Gentilini\inst{2}\thanks{Corresponding author}}

\institute{{RMIT University, School of Science, Melbourne, Australia \and Dip. di Matematica e Informatica, Universit\`a di Perugia, Via Vanvitelli 1, Perugia (IT) .\\
$\{$marc.demange$|$david.ellison$\}$@rmit.edu.au,
$\{$raffaella.gentilini$\}$@dmi.unipg.it}}

\begin{document}
\maketitle

\begin{abstract}
The Firefighter Problem ({\sc FP}) is a graph problem originally introduced in 1995 to model the spread of a fire in a graph,  which has  attracted considerable attention in the literature. The goal is to devise a strategy to employ a given sequence of firefighters on strategic  points in the graph  in order to  contain efficiently the fire (which spreads  from each unprotected vertex to all of it neighbours on  successive time steps). 

Recently,  an \emph{online} version of {\sc FP}--- where  the number of firefighters available at each turn are revealed in \emph{real-time}--- has been introduced in \cite{Demange2018,Demange2019} and studied on trees. 
In this paper, we extend the work in  \cite{Demange2018,Demange2019} by considering the online containment of fire on square grids.  In particular, we provide a set of sufficient conditions that allow to   solve the online  version of the firefighting problem on infinite square grids, illustrating the corresponding  fire containment strategies.
\end{abstract}

\section{Introduction}

The Firefighter Problem ({\sc FP}, from now on)  is a   combinatorial problem  introduced  by Bert Hartnell in 1995 \cite{Hartnell95}, providing a  deterministic, discrete-time model of the spread of a fire on the vertices of a graph. Suppose that a fire  breaks out at time $0$ at a vertex $v$ of a graph $G$. 
At each subsequent time $t$, $f_t$  firefighters  protect a corresponding number of $f_t$ vertices  in  $G$, and then   the fire spreads from each burning vertex to all of its undefended neighbours. Once a vertex is burning or defended, it remains so from then onwards. The process terminates when the fire can no longer spread.  In the case of finite graphs, the aim is to save as many vertices as possible, while in the infinite case, the aim is that of simply  containing the fire.

Since its introduction in  \cite{Hartnell95},   {\sc FP}   has been studied intensively in the literature \cite{AdjiashviliBZ17,Anshelevich2012,Finbow2007,FinbowM09,Garcia2015}. In particular  {\sc FP}  has been shown  NP-complete for bipartite graphs \cite{Hartnell95}.    Finbow et al. have strengthened this result  \cite{Finbow2007},  proving that   {\sc FP} is NP-complete even if restricted to trees with maximum degree three. In contrast, it is solvable in polynomial time for graphs of maximum degree three, if the fire starts at a vertex of degree two. A polynomial time approximation scheme for {\sc FP} on trees has been recently provided in \cite{AdjiashviliBZ17}.  Beside trees,  {\sc FP} has been extensively studied on the families of graphs of grids \cite{Fomin2016,messinger07,messinger08,raff08,wang}.
Wang and Moeller \cite{wang} proved that one firefighter per turn cannot  control a fire sourcing from a vertex $v\in \mathbb{L}_2=\mathbb{Z}\times \mathbb{Z}$, while  two firefighters per turn are sufficient to solve  {\sc FP}  on $\mathbb{L}_2$   within 8 turns and  18 burnt vertices. Ng and Raff \cite{raff08} proved that any periodic function $(f_t)_{t\geq 1}$ whose average exceeds $\frac{3}{2}$ allows the firefighters to control any finite-source fire in $\mathbb{L}_2$.
 The interested reader can refer to \cite{FinbowM09,Fomin2016,Garcia2015} for a survey of recent results on the complexity analysis of {\sc FP} and its variants. 
  
  Recently, Coupechoux et al. \cite{Demange2018,Demange2019} have considered an \emph{online} version of    {\sc FP}, where the number of firefighters available at each turn are revealed in \emph{real-time}. In \cite{Demange2018,Demange2019}, the structure of the underlying graph in the online {\sc FP} is a tree, and  suitable competitivity results are provided. In this work, we consider the online version of  {\sc FP} on the  infinite Cartesian  grid $\mathbb{L}_2=\mathbb{Z}\times \mathbb{Z}$.

\section{Offline vs Online Firefighting on Grids}
The classic offline version of the firefighter problem can be understood as a deterministic one-player game, where Player $1$ knows in advance the sequence $(f_i)_{i\geq 1}$ of firefighters available overall the game. In contrast, within the online version of the firefighter problem, an adversary called Player $2$ reveals   to Player $1$---turn by turn---how many firefighters are ready to be used in the current turn of the game.  Therefore, the online firefighter problem can be understood as a two player game. More precisely, an instance of the online firefighter problem is given by the tuple $(\mathcal{A}, (f_i)_{i\geq 1})$, where both players are aware of the arena of the game, $\mathcal{A}=\langle G,v\rangle$, which  is composed  by  the graph $G$\-- the Cartesian grid $\mathbb{L}_2=\mathbb{Z}\times \mathbb{Z}$ in this paper\-- and the ignition vertex $v\in \mathbb{L}_2$. Instead,   only Player $2$ knows  the firefighters sequence $(f_i)_{i\geq 1}$. At each turn $i$ of the game, Player $2$ reveals $f_i$ to Player $1$. Then, Player $1$ chooses $m\leq f_i$ vertices (neither protected nor burned), where to place a new firefighter. Finally, the fire spreads on each unprotected  neighbour of a vertex on fire, leading to the next turn of the game. Player $2$ wins if at each turn a new vertex is burning, otherwise Player $1$ wins.  

We provide a constraint applying to the sequence of firefighters $(f_i)_{i\geq 1}$ (cf. Condition \eqref{eq1}) that can be shown to be a sufficient condition for Player $1$ to win the offline version of the firefighter problem, while there is an instance of the online firefighter problem that fulfills Condition \eqref{eq1}  where Player $1$ looses.

\begin{equation}\label{eq1}
\exists N\geq 1: \sum_{i=1}^N f_i \geq 4N
\end{equation}

Intuitively, Condition \eqref{eq1} guarantees the existence of a turn of the game for which the \emph{global} number of  firefighters deployed (from the beginning of the game up to the current turn) have been $4$ times the number of turns played so far. Hence, an offline strategy for Player $1$ can  rely on the knowledge of $(f_i)_{i\geq 1}$ to surround the fire at the right distance. However, in the online version of the game, Player $1$ does not have any information  about the turn of the game fulfilling Condition \eqref{eq1}.  This is formalized by Theorem \ref{FPgridOffineOnline}, below.
\begin{theorem}\label{FPgridOffineOnline}
Condition \eqref{eq1} is sufficient (resp. not sufficient) for Player $1$ to win the offline (resp. online) version of the firefighter problem. Moreover, $5$ turns are enough to make any online strategy for Player $1$ to fail.

\end{theorem}
\begin{proof}
Wlog, suppose that the ignition vertex is $(0,0)$ and let $N$ be the smallest index such that  $\sum_{i=1}^N f_i \geq 4N$. An offline winning strategy  for Player $1$ is the following: build a diamond-shape encirclement at distance $N$ from the ignition vertex, by placing each firefighter eventually available on a grid point $(x,y)$ satisfying $|x|+|y|=N$. There are exactly $4N$ grid vertices $(x,y)$ such that $|x|+|y|=N$. Hence, Condition \eqref{eq1} guarantees that at turn $N$ the fire gets completely encircled. 

Being unaware of $N$, Player $1$ cannot apply the above strategy in an online set-up. We show that the adversary has a winning strategy $(f_i)_{i\geq 1}$, where $(f_i)_{i\geq 1}$ satisfies Condition \eqref{eq1}. Given $j>1$, denote by  $(f^j_i)_{i\geq 1}$  the firefighter sequence where $f_1=1$, $f_j=4j-1$ and $f_p=0$ for $p\notin\{1,j\}$.
Let $v\in D_\ell$, for some positive integer $\ell$, be the first node protected by Player $1$, where $D_\ell$ is the set of vertices in the grid $\mathbb{Z}\times\mathbb{Z}$ at distance $\ell$ from the ignition vertex. A winning strategy for the adversary 
is the following. If $\ell\leq 2$ (i.e. $v$ is at distance $1$ or $2$ from the ignition vertex), then the adversary provides the firefighter sequence $(f_i^5)_{i\geq 1}$. Hence, at the end of turn $4$, $20$ firefighters are needed to surround the fire, while $f_5=19$. Otherwise, if $\ell>2$ (i.e. $v$ is at distance  greater than $2$ from the ignition vertex), the adversary provides the sequence of firefighters $(f^2_i)_{i\geq 1}$. Hence, at the end of turn $2$, there will be an unprotected node $u\in D_2$  next to the fire  at distance $2$ from the ignition vertex. 
So, in both cases, the firefighter sequence is of the form $(f_i^j)_{i\geq 1}$ for some $j$ and the fire is not encircled after turn $j$. Since no firefighter will be available for the rest of the game, the fire will escape.
\end{proof}

\section{Online Containment of Fires on Grids}

In the previous section, we have considered sequences of firefighters $(f_i)_{i\geq 1}$ satisfying    Condition \eqref{eq1}, showing that such a condition is sufficient for Player $1$ to win offline, while there are instances of the online firefighter problem fulfilling  Condition\eqref{eq1} where Player $1$ loses.  The purpose of this section is that of providing sufficient conditions toward the \emph{online containment} of fires  on grids. Our first result (cf. Subsection \ref{1firefighterSec} below) shows that if $(f_i)_{i\geq 1}$ fulfils Condition \eqref{eq1} and at least one firefighter is eventually always available, then Player $1$ has a strategy to win online. The following Subsection \ref{SecConstraintOnline} considers the problem of weakening Condition \eqref{eq1} in order to define further sufficient conditions for the online containment of fires on grids.

\subsection{At Least  One Firefighter Always Available}\label{1firefighterSec}
Suppose that the sequence of firefighters revealed by Player $2$ is such that eventually at least one firefighter will be  available on each turn of the game. Then, we show that Condition \eqref{eq1} becomes sufficient for Player $1$ to win online. Intuitively, this is because the firefighter(s)  available at each turn can be employed \emph{next} to the fire to build incrementally a tight
 encirclement of it, waiting for later reinforcement. This way, no firefighter is wasted during subsequent turns of the game, while Condition \eqref{eq1} guarantees that eventually, there will be enough firefighters to close the  encirclement surrounding the fire. 
 
  More precisely, consider the simpler scenario where Player $1$ receives \emph{exactly} one firefighter at each turn $1,2,\dots ,\mu-1$ and  $m$ firefighters at turn $\mu$, where $m\geq  4\mu - (\mu-1)$. 
  For instance, Figure \ref{figu1} considers the case where $\mu=4$, $f_1=f_2=f_3=1$ and  $f_4=13$ and illustrates an online winning strategy for Player $1$.   
Such a strategy works as follows: Player $1$ uses the available firefighter at each turn to build two diagonal walls (cf. the positioning of the only firefighter received for the first three turns). This way, at the beginning of each turn $1\leq i\leq \mu$, the fire is always contained within a perimeter of size $4i$,  and each firefighter previously employed at some turn $j<i$ has been placed on such a perimeter. Therefore, the encirclement of the fire can be completed as soon as Condition \eqref{eq1} is fulfilled (at turn $4$, for the instance illustrated in  Figure \ref{figu1}). 
 \begin{center}
 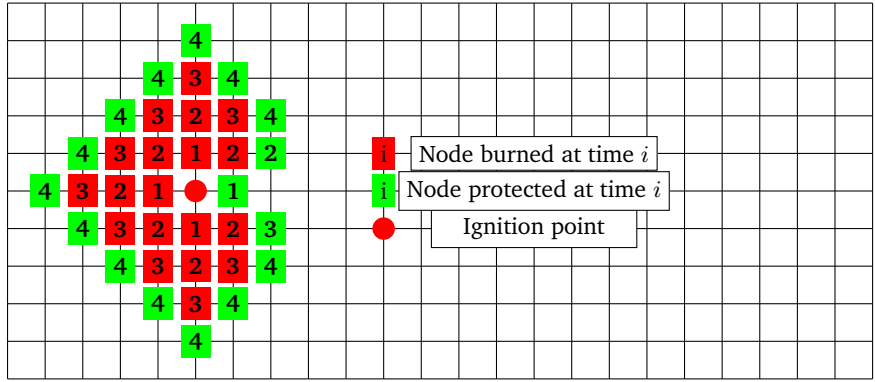
\begin{figure}[ht!]
 \begin{tikzpicture}
    [
     box/.style={rectangle, minimum size=0.1cm},
    ]

\foreach \x in {0,0.5,...,11.5}{
    \foreach \y in {-0.5,0,0.5,...,4,4.5}
        \node at (\x,\y){};
}

\foreach \x in {0,0.5,...,11}{
    \foreach \y in {-0.5,0,0.5,...,4,4.5}
        \draw (\x,\y) --  (\x+0.5,\y);
}
\foreach \x in {0,0.5,...,11.5}{
    \foreach \y in {-0.5,0,0.5,...,4}
        \draw (\x,\y) --  (\x,\y+0.5);
}

\node[circle,fill=red  ] at (2.5,2){};  
\node[box,fill=green  ] at (3,2){\small{\bf 1}};  
\node[box,fill=red  ] at (2,2){\small{\bf  1}}; 
\node[box,fill=red  ] at (2.5,1.5){\small{\bf 1}};   
\node[box,fill=red  ] at (2.5,2.5){\small{\bf 1}};   
\node[box,fill=green  ] at (3.5,2.5){\small{\bf 2}};  

\node[box,fill=red  ] at (3,2.5){\small{\bf 2}};  
\node[box,fill=red  ] at (2.5,3){\small{\bf 2}};  
\node[box,fill=red  ] at (2,2.5){\small{\bf 2}};  
\node[box,fill=red  ] at (1.5,2){\small{\bf 2}};  
\node[box,fill=red  ] at (2,1.5){\small{\bf 2}};  
\node[box,fill=red  ] at (2.5,1){\small{\bf 2}};  
\node[box,fill=red  ] at (3,1.5){\small{\bf 2}};  

\node[box,fill=green  ] at (3.5,1.5){\small{\bf 3}}; 
\node[box,fill=red  ] at (3,1){\small{\bf 3}}; 
\node[box,fill=red  ] at (2.5,0.5){\small{\bf 3}}; 
\node[box,fill=red  ] at (2,1){\small{\bf 3}}; 
\node[box,fill=red  ] at (1.5,1.5){\small{\bf 3}}; 
\node[box,fill=red  ] at (1,2){\small{\bf 3}};
\node[box,fill=red  ] at (1.5,2.5){\small{\bf 3}};  
\node[box,fill=red  ] at (2,3){\small{\bf 3}};  
\node[box,fill=red  ] at (2.5,3.5){\small{\bf 3}};  
\node[box,fill=red  ] at (3,3){\small{\bf 3}};  

\node[box,fill=green  ] at (3.5,1){\small \bf 4}; 
\node[box,fill=green  ] at (3,0.5){\small \bf 4};
\node[box,fill=green  ] at (2.5,0){\small \bf 4};
\node[box,fill=green  ] at (2,0.5){\small \bf 4};
\node[box,fill=green  ] at (1.5,1){\small \bf 4};    
\node[box,fill=green  ] at (1,1.5){\small \bf 4};
\node[box,fill=green  ] at (0.5,2){\small \bf 4};  
\node[box,fill=green  ] at (1,2.5){\small \bf 4};
\node[box,fill=green  ] at (1.5,3){\small \bf 4};
\node[box,fill=green  ] at (2,3.5){\small \bf 4};
\node[box,fill=green  ] at (2.5,4){\small \bf 4};  
\node[box,fill=green  ] at (3,3.5){\small \bf 4};
\node[box,fill=green  ] at (3.5,3){\small \bf 4};       

\node[circle,fill=red  ] at (5, 1.5){};    
\node[draw,align=left,fill=white] at (7,1.5){\small{$\quad$Ignition point$\quad$}}; 
\node[box,fill=red  ] at (5, 2.5){\small{i}};    
\node[draw,align=left,fill=white] at (7,2.5){\small{Node burned at time $i$}}; 
\node[box,fill=green  ] at (5, 2){\small{i}};    
\node[draw,align=left,fill=white] at (7,2){\small{Node protected at time $i$}};               

\end{tikzpicture} 
\caption{Online strategy to encircle the fire with (exactly) one firefighter available at each turn of  the game $1,\dots, N-1$, until   Condition \eqref{eq1} gets fulfilled at turn $N$.}\label{figu1}
\end{figure}
\end{center}  
  The general case, where Player $1$ eventually receives \emph{at least} one firefighter on each  turn  until Condition \eqref{eq1} is accomplished, is slightly more in involved. Roughly,  it  is solved as follows. As far as Player $1$ receives $f_j>1$ firefighters at each turn $j$ (while Condition \eqref{eq1} still needs to be accomplished) he will place the guaranteed available firefighter on a diagonal wall in front of the advance of the fire, while the extra firefighters (out from the guaranteed one) will be employed to encircle the fire, waiting for later reinforcement. As soon as Condition \eqref{eq1} is satisfied, the encirclement will be completed. Example \ref{example1} below  gives more details on the above sketched winning online strategy for Player $2$. This leads  to the results in Theorem \ref{theo2}, below.
  \begin{theorem}\label{theo2} Let the sequence of firefighters $(f_i)_{i\geq 1}$ revealed by Player $2$   be consistent with 
  Condition \eqref{eq1} and suppose that there exists an index $M$ such that $f_i\geq 1$ if $i\geq M$, and $f_i=0$ otherwise.
 Then, Player $1$ has an online strategy to win the firefighter problem on grids.
 
 \end{theorem}
 \begin{proof}
  Let $N$ be the smallest  index such that $\sum_{i=1}^N f_i\geq  4N$ and assume wlog that the ignition vertex is  $v=(0,0)$. 
 
 By hypothesis, the fire spreads uncontrolled for $M-1$ turns. Hence, when the first firefighter(s) appear at turn $M$ the fire has a diamond shape, burning each grid  vertex $(x,y)$ within the polygon enclosed by the lines $y=mx+k, |m|=1, |k|=M$, i.e. each grid vertex satisfying $|x|+|y|<M$. Therefore, if $M=N$, then Player $1$ can tightly encircle the fire  by 
 assigning a firefighter to each grid-vertex $(x,y)$ such that $|x|+|y|=N$ (there are $4N$ such vertices).
 
Otherwise ($M<N$) , we first consider a simpler scenario where Player $1$ receives \emph{exactly} one firefighter from turn $M$ till turn $N-1$, providing an online strategy  to contain the fire. We will then  generalize such a winning strategy for Player $1$ to deal with the general case where $f_i\geq 1$ for $M\leq i\leq N-1$.
Having one firefighter available at each turn $M\leq i\leq N-1$, Player $1$ can progressively build two diagonal walls next to the fire: the first (resp. second) wall  from $(M,0)$ along the semi-line $y=x-M,x\geq M$ (resp. $y=-x+M,x\geq M$).  Precisely, at each turn $i=M+j$, for $0\leq j<N-M$, Player $1$ protects the vertex $(M+\lceil \frac{j}{2}\rceil, (-1)^j\lfloor  \frac{j}{2}\rfloor) $, alternating the building of such walls (cf. Figure \ref{figu1}).

Therefore, at the end of each turn  $M\leq i < N$ the fire burns  each grid vertex internal  to the the polygon defined by the  following inequalities (cf. Figure \ref{figu1}): 
\begin{equation}\label{eq3} 
 \begin{cases}
          y\leq x+i\\ y\geq -x-i \\ y\geq x-i \\ y\leq -x+i \\ y\geq x-M \\ y\leq -x+M
       \end{cases}
\end{equation}
The number of grid vertices on the perimeter of such a polygon is exactly $4i$. Moreover, we have placed all the firefighters globally received (from turn $M$ till turn $i$) on such a perimeter, along the diagonal walls defined by $ y= x-M, y = -x+M$. Hence, at turn $N$ we can completely surround the fire (cf. Figure \ref{figu1}).

We now turn on considering the general case where we receives at least one firefighter from turn $M$ to turn $N>M$ (rather than having $f_i=1$ for $M\leq i <N$). Note that, having $f_i=1$ for $M\leq i <N$ allowed us to place each new arriving firefighter on the perimeter of the polygon induced by the set of inequalities in  \ref{eq3} along the facets $y= x-M, y=-x+M$. Moreover, such facets are built incrementally  so that at the end of each turn $i$ the firefighters on the field are \emph{next} to the fire. This ensures that such firefighters will be useful to contain the fire for the rest of the game. To maintain such a property in the general case (namely, $f_i\geq 1$ for $M\leq i <N$) we proceed as follows. As far as we get $1$ firefighter we keep on building the two diagonal walls along $y= x-M, y=-x+M$. As soon as  we receive strictly more than one firefighter, say at turn $P\geq M$, we use the extra $f_P-1$ firefighters  to surround  the fire. More precisely, we stop building one of the two walls, say the one along $y=y-M$, and start  enqueuing the   $f_P-1$ extra firefighters along the perimeter of the burning polygon (starting e.g. from the facet defined by $y=-x+P$). At turn $P+1$ we will start building a new diagonal wall from the last enqueued firefighter, so that we always have two diagonal walls on construction (orthogonal to two advancing  fronts of fire). We keep proceeding as above, i.e. building two diagonal walls blocking the fire as far as we get one firefighter, and using the extra firefighters to surround the fire as soon as we get strictly more than one firefighter (cf. Figure   \ref{figurefig2}). This way, at each turn $i=M\dots N-1$, the fire is contained within a polygon whose perimeter has at most $4i$ grid vertices, and each firefighter used so far in the game is protecting  such a perimeter. Therefore, by hypothesis  the firefighters available at turn   $N$ will be  enough to close the tight  partial encirclement of the fire built at previous turns. 
 \end{proof}
   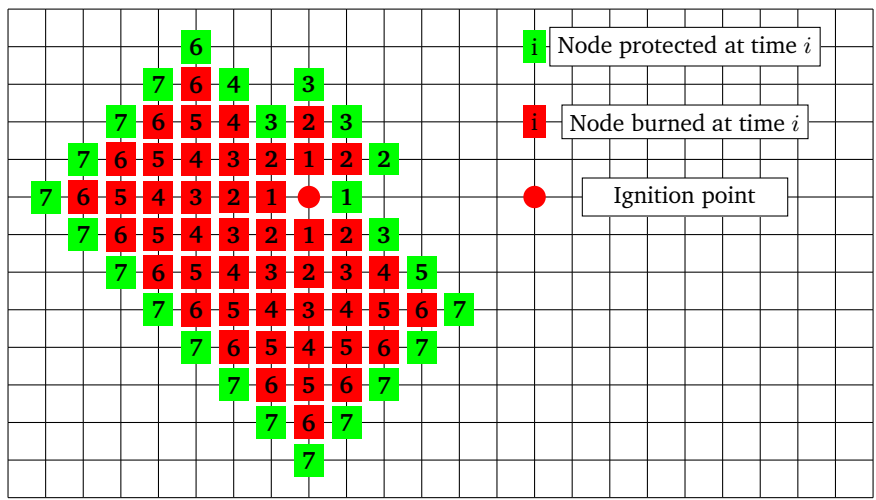
\begin{figure}[t!]
 
 \centering
 
 \begin{tikzpicture}
    [
     box/.style={rectangle, minimum size=0.1cm},
    ]

\foreach \x in {0,0.5,...,11.5}{
    \foreach \y in {1,1.5,...,7.5}
        \node at (\x,\y){};
}

\foreach \x in {0,0.5,...,11}{
    \foreach \y in {1,1.5,...,7.5}
        \draw (\x,\y) --  (\x+0.5,\y);
}
\foreach \x in {0,0.5,...,11.5}{
    \foreach \y in {1,1.5,...,7}
        \draw (\x,\y) --  (\x,\y+0.5);
}

\node[circle,fill=red  ] at (4,5){};  

\node[box,fill=green  ] at (4.5,5){\small{\bf 1}};  
\node[box,fill=red  ] at (3.5,5){\small{\bf  1}}; 
\node[box,fill=red  ] at (4,5.5){\small{\bf 1}};   
\node[box,fill=red  ] at (4,4.5){\small{\bf 1}};

\node[box,fill=green  ] at (5,5.5){\small{\bf 2}};  
\node[box,fill=red  ] at (4.5,5.5){\small{\bf 2}};  
\node[box,fill=red  ] at (4,6){\small{\bf 2}};  
\node[box,fill=red  ] at (3.5,5.5){\small{\bf 2}}; 
\node[box,fill=red  ] at (3,5){\small{\bf 2}};
  \node[box,fill=red  ] at (3.5,4.5){\small{\bf 2}};  
  \node[box,fill=red  ] at (4,4){\small{\bf 2}};  
  \node[box,fill=red  ] at (4.5,4.5){\small{\bf 2}};

\node[box,fill=green  ] at (5,4.5){\small{\bf 3}}; 
\node[box,fill=green  ] at (4.5,6){\small{\bf 3}};  
\node[box,fill=green  ] at (4,6.5){\small{\bf 3}};  
\node[box,fill=green  ] at (3.5,6){\small{\bf 3}};  
\node[box,fill=red  ] at (3,5.5){\small{\bf 3}};  
\node[box,fill=red  ] at (2.5,5){\small{\bf 3}};  
\node[box,fill=red  ] at (3,4.5){\small{\bf 3}};  
\node[box,fill=red  ] at (3.5,4){\small{\bf 3}};  
\node[box,fill=red  ] at (4,3.5){\small{\bf 3}};  
\node[box,fill=red  ] at (4.5,4){\small{\bf 3}};  

\node[box,fill=green  ] at (3,6.5){\small{\bf 4}};  
\node[box,fill=red  ] at (3,6){\small{\bf 4}}; 
\node[box,fill=red  ] at (2.5,5.5){\small{\bf 4}}; 
\node[box,fill=red  ] at (2,5){\small{\bf 4}}; 
\node[box,fill=red  ] at (2.5,4.5){\small{\bf 4}}; 
\node[box,fill=red  ] at (3,4){\small{\bf 4}};
\node[box,fill=red  ] at (3.5,3.5){\small{\bf 4}};  
\node[box,fill=red  ] at (4,3){\small{\bf 4}};
\node[box,fill=red  ] at (4.5,3.5){\small{\bf 4}};   
\node[box,fill=red  ] at (5,4){\small{\bf 4}};   

\node[box,fill=green ] at (5.5,4){\small{\bf 5}};   
\node[box,fill=red ] at (5,3.5){\small{\bf 5}}; 
\node[box,fill=red ] at (4.5,3){\small{\bf 5}};  
\node[box,fill=red ] at (4,2.5){\small{\bf 5}};  
\node[box,fill=red ] at (3.5,3){\small{\bf 5}}; 
\node[box,fill=red ] at (3,3.5){\small{\bf 5}}; 
\node[box,fill=red ] at (2.5,4){\small{\bf 5}};   
\node[box,fill=red ] at (2,4.5){\small{\bf 5}}; 
\node[box,fill=red ] at (1.5,5){\small{\bf 5}};  
\node[box,fill=red ] at (2,5.5){\small{\bf 5}}; 
\node[box,fill=red ] at (2.5,6){\small{\bf 5}}; 

\node[box,fill=green ] at (2.5,7){\small{\bf 6}};   
\node[box,fill=red ] at (2.5,6.5){\small{\bf 6}};
\node[box,fill=red ] at (2,6){\small{\bf 6}}; 
\node[box,fill=red ] at (1.5,5.5){\small{\bf 6}};  
\node[box,fill=red ] at (1,5){\small{\bf 6}};   
\node[box,fill=red ] at (1.5,4.5){\small{\bf 6}}; 
\node[box,fill=red ] at (2,4){\small{\bf 6}};  
\node[box,fill=red ] at (2.5,3.5){\small{\bf 6}};     
\node[box,fill=red ] at (3,3){\small{\bf 6}};    
\node[box,fill=red ] at (3.5,2.5){\small{\bf 6}};   
\node[box,fill=red ] at (4,2){\small{\bf 6}};   
\node[box,fill=red ] at (4.5,2.5){\small{\bf 6}};
\node[box,fill=red ] at (5,3){\small{\bf 6}};
\node[box,fill=red ] at (5.5,3.5){\small{\bf 6}};     

\node[box,fill=green ] at (2,6.5){\small{\bf 7}};  
\node[box,fill=green ] at (1.5,6){\small{\bf 7}};  
\node[box,fill=green ] at (1,5.5){\small{\bf 7}};  
\node[box,fill=green ] at (0.5,5){\small{\bf 7}};
\node[box,fill=green ] at (1,4.5){\small{\bf 7}};  
\node[box,fill=green ] at (1.5,4){\small{\bf 7}};  
\node[box,fill=green ] at (2,3.5){\small{\bf 7}};  
\node[box,fill=green ] at (2.5,3){\small{\bf 7}};  
\node[box,fill=green ] at (3,2.5){\small{\bf 7}};  
\node[box,fill=green ] at (3.5,2){\small{\bf 7}};
\node[box,fill=green ] at (4,1.5){\small{\bf 7}}; 
\node[box,fill=green ] at (4.5,2){\small{\bf 7}};  
\node[box,fill=green ] at (5,2.5){\small{\bf 7}};  
\node[box,fill=green ] at (5.5,3){\small{\bf 7}};  
\node[box,fill=green ] at (6,3.5){\small{\bf 7}};

\node[circle,fill=red  ] at (7, 5){};    
\node[draw,align=left,fill=white] at (9,5){\small{$\quad$Ignition point$\quad$}}; 
\node[box,fill=red  ] at (7, 6){\small{i}};    
\node[draw,align=left,fill=white] at (9,6){\small{Node burned at time $i$}}; 
\node[box,fill=green  ] at (7, 7){\small{i}};    
\node[draw,align=left,fill=white] at (9,7){\small{Node protected at time $i$}};               

\end{tikzpicture} 

 \caption{Winning the fire online  with  at least  one firefighter eventually always available until   Condition \eqref{eq1} is satisfied.}\label{figurefig2}
\end{figure}
  \begin{example}\label{example1}
 
Let the sequence of firefighters $(f_i)_{i\geq 1}$ received by Player $1$ be such that $f_1=f_2=1$, $f_3=4$, $f_4=f_5=f_6=1$, and $f_7=15$. Therefore, Player $1$ receives one \emph{or more} firefighters for the first six turns and a number of firefighters leading to the fulfillment of Condition \eqref{eq1} on the seventh turn, i.e. $\sum_{i=1}^7 f_i \geq 28$.

As illustrated in Figure \ref{figurefig2}, to win online, 
 Player $1$ proceeds as follows. As far as he receives exactly one firefighter per turn, he  employs them to build two diagonal walls (cf. the positioning of the first two firefighters)  against the front of the fire. In our example, this happens for the first two turns. At the third turn Player $1$ receives $4$ firefighters, i.e. a number of firefighter that is strictly greater than one but not enough to let Condition \eqref{eq1} being satisfied. Then, Player 1 uses one firefighter on the diagonal walls and the extra three firefighters to start surrounding the fire (cfr  the positioning of the firefighters received at turn $3$ in Figure \ref{figurefig2}). At the next turn\--i.e. at turn $4$ in our 
\-- Player $1$ will start building a new diagonal wall from the last enqueued firefighter, so that he  always have two diagonal walls under construction, orthogonal to two advancing  fronts of fire. Such diagonal walls will be alternatively enlarged (at turn $5$ and $6$ in our example), until  Player $1$ will be able to completely surround the fire (at turn $7$ in our example),  when the available firefighters will allow to have Condition \eqref{eq1} satisfied. 
\end{example}

 \subsection{Sufficient Conditions to Win the Fire Online on Grids}\label{SecConstraintOnline}

Consider the following natural generalisation of Condition \eqref{eq1}:

\begin{equation}\label{eq2}
\exists N\geq 1: \sum_{i=1}^N f_i\geq \ell\cdot N
\end{equation}
 We  study for which $\ell\in\mathbb{N}$ Condition \eqref{eq2} guarantees to extinguish 
   the fire online.  
  
 In the offline case, $\ell=4$ is sufficient  to contain the fire as stated in Theorem~\ref{FPgridOffineOnline}. 
 As already noted, such a strategy does not work online, since Player $1$ is not aware of how many firefighters will be available at each turn of the game. Therefore, the attempt of building a diamond-shaped encirclement at distance $N$ could result into a waste of all the firefighters placed on such an encirclement, as soon as the latter gets broken by the spreading fire.

 However, one can notice that 
 in this case a maximum number of $4N$ firefighters are lost. Based on this observation, it is possible to come up with an online strategy for Player $1$ that allows him to contain the fire under Condition \eqref{eq2} for $\ell=16$. Such a strategy  proceeds as follows: as soon as the encirclement under construction gets broken at turn $t$ of the game, Player $1$ starts to build a new diamond-shaped encirclement at distance $2t$ from the 
 ignition point.
Let $M$ be the turn of the game at which Condition \eqref{eq2} is satisfied for $\ell=16$, i.e. assume that a global number of $16M$ firefighters gets available overall the first $M$ turns. We show that at  turn $M$, we are guaranteed to have enough firefighters to complete an encirclement at distance $2M$. In fact, $4*2M=8M$ firefighters are needed for such an encirclement, while the number of firefighters lost in previous turns is bounded by $4M*\sum_{i=1}^\infty {\frac{1}{2}}^i = 8M$.

Therefore, we obtain:

 \begin{theorem}\label{theo3} Let the sequence of firefighters $(f_i)_{i\geq 1}$ revealed by Player $2$ to Player $1$   be such that:
 $$\exists N\geq 1: \sum_{i=1}^N f_i\geq 16\cdot N$$
 Then, Player $1$ has an online strategy to win the firefighter problem on grids.
 
 \end{theorem}

\end{document}